\documentclass[sigconf, authorversion, nonacm]{acmart}

\usepackage[utf8]{inputenc}

\acmYear{2022}
\setcopyright{acmlicensed}\acmConference[SPAA '22]{Proceedings of the 34th ACM Symposium on Parallelism in Algorithms and Architectures}{July 11--14, 2022}{Philadelphia, PA, USA}
\acmBooktitle{Proceedings of the 34th ACM Symposium on Parallelism in Algorithms and Architectures (SPAA '22), July 11--14, 2022, Philadelphia, PA, USA}
\acmPrice{15.00}
\acmDOI{10.1145/3490148.3538573}
\acmISBN{978-1-4503-9146-7/22/07}



\usepackage[normalem]{ulem} 
\usepackage{amsmath}

\usepackage{amsthm}
\usepackage{amsfonts}
\usepackage{xcolor}
\usepackage{enumitem}
\usepackage{xspace}
\usepackage[normalem]{ulem} 

\definecolor{darkgreen}{rgb}{0,0.5,0}
\usepackage{hyperref}
\hypersetup{
    unicode=false,          
    colorlinks=true,        
    linkcolor=blue,          
    citecolor=darkgreen,        
    filecolor=magenta,      
    urlcolor=cyan           
}
\usepackage[capitalize,nameinlink]{cleveref}
\crefformat{footnote}{#2\footnotemark[#1]#3}

\newcommand{\fabian}[1]{}
\newcommand{\mab}[1]{}
\newcommand{\mabinline}[1]{}
\newcommand{\john}[1]{}
\newcommand{\seth}[1]{}
\newcommand{\ftodo}[1]{}
\newcommand{\fixme}[1]{}
\newcommand{\TODO}{}

\renewcommand{\paragraph}[1]{\vspace{0.09in}\noindent{\bf \boldmath #1}} 

\newtheorem{theorem}{Theorem}

\newtheorem{lemma}[theorem]{Lemma}
\newtheorem{corollary}[theorem]{Corollary}

\newtheorem{definition}{Definition}

\newcommand{\defn}[1]{{\textit{\textbf{\boldmath #1}}}}

\renewcommand{\Pr}{\mathbb{P}}

\newcommand{\threshold}{\kappa}

\newcommand{\poly}[1]{\text{\rm poly}\ensuremath{(#1)}}
\newcommand{\pmin}{p_{\rm min}}
\newcommand{\sysname}{Decodable Backoff Algorithm\xspace}

\newcommand{\secref}[1]     {Section~\ref{sec:#1}}

\newcommand{\thmref}[1]     {Theorem~\ref{thm:#1}}

\newcommand{\lemref}[1]     {Lemma~\ref{lem:#1}}
\newcommand{\lemreftwo}[2]  {Lemmas \ref{lem:#1} and~\ref{lem:#2}}

\newcommand{\corref}[1]     {Corollary~\ref{cor:#1}}
\renewcommand{\eqref}[1]      {Equation~(\ref{eq:#1})}

\newcommand{\defref}[1]     {Definition~\ref{def:#1}}

\renewcommand{\epsilon}{\varepsilon}

\begin{document}

\title{Contention Resolution for Coded Radio Networks}

\author{Michael A. Bender}
\affiliation{%
  \institution{Stony Brook University}
  \city{Stony Brook}
  \state{New York}
  \country{USA}
}
\email{bender@cs.stonybrook.edu}

\author{Seth Gilbert}
\affiliation{%
  \institution{National University of Singapore}
  \country{Singapore}
}
\email{seth.gilbert@comp.nus.edu.sg}

\author{Fabian Kuhn}
\affiliation{%
  \institution{University of Freiburg}
  \city{Freiburg}
  \country{Germany}
}
\email{kuhn@cs.uni-freiburg.de}

\author{John Kuszmaul}
\affiliation{%
  \institution{Yale University}
  \city{New Haven}
  \state{Connecticut}
  \country{USA}
}
\email{john.kuszmaul@yale.edu}

\author{Muriel M\'edard}
\affiliation{%
  \institution{Massachusetts Institute of Technology}
  \city{Cambridge}
  \state{Massachusetts}
  \country{USA}
}
\email{medard@mit.edu}



\begin{abstract}
Randomized backoff protocols, such as exponential backoff, are a powerful tool for managing access  to  a  shared  resource, often a wireless communication channel (e.g.,~\cite{802.11-standard}). For a wireless device to transmit successfully, it uses a backoff protocol to ensure exclusive access to the channel. Modern radios, however, do not need exclusive access to the channel to communicate; in particular, they have the ability to receive useful information even when more than one device transmits at the same time. These capabilities have now been exploited for many years by systems that rely on interference cancellation, physical layer network coding and analog network coding to improve efficiency. For example, Zigzag decoding~\cite{GollakotaKa08} demonstrated how a base station can decode messages sent by multiple devices simultaneously. 

In this paper, we address the following question: \emph{Can we design a backoff protocol that is better than exponential backoff when exclusive channel access is not required}. We define the Coded Radio Network Model, which generalizes traditional radio network models (e.g.,~\cite{ChlamtacKu85}). We then introduce the \sysname, a randomized backoff protocol that achieves an optimal throughput of $1-o(1)$. (Throughput $1$ is optimal, as simultaneous reception does not increase the channel capacity.)  The algorithm breaks the constant throughput lower bound for traditional radio networks~\cite{GoldbergJeKa00,GoldbergJeKa04,Goldberg-notes-2000}, showing the power of these new hardware capabilities.
\end{abstract}

\maketitle

\sloppy


\section{Introduction}
\label{sec:intro}
 
Randomized \defn{backoff} protocols~\cite{Goldberg-notes-2000,HastadLeRo96,MetcalfeBo76,Abramson70, AbramsonKu73,802.11-standard},
such as exponential backoff~\cite{MetcalfeBo76}, are a powerful tool for managing access to a shared resource, often a communication channel.
These protocols are used prominently in wireless networks from AlohaNet~\cite{Abramson70, AbramsonKu73,Binder75,roberts:aloha} to 802.11~\cite{802.11-standard} and beyond~\cite{CasiniGH07}.
Wireless devices typically need to transmit their messages to a base station, and in order for a device's transmission to be successful, the device must have \defn{exclusive access} to the channel, that is, no other wireless device can be transmitting.
If more than one device transmits simultaneously, there is a \defn{message collision} and the base station receives indecipherable noise~\cite{Goldberg-notes-2000,kurose:computer,xiao:performance,DBLP:journals/jacm/GoldbergMPS00,10.1145/28395.28422,MetcalfeBo76}.

The problem of designing backoff protocols is traditionally formalized as the \defn{contention-resolution} problem.  The devices correspond to  agents that arrive over time, each with a message that needs to be transmitted.  To perform the transmission, the device requires exclusive  access to the channel~\cite{ChlamtacKu85,Goldberg-notes-2000,kurose:computer,xiao:performance,DBLP:journals/jacm/GoldbergMPS00,10.1145/28395.28422,MetcalfeBo76,Abramson70, AbramsonKu73}. 
Time is subdivided into synchronized \defn{slots}, which are sized to fit a message, and the objective is to maximize the \defn{utilization} or \defn{throughput} of the channel, which is roughly (but not exactly) the fraction of slots that successfully transmit messages.
For example, traditional binary exponential backoff yields $\Theta(1/\log{n})$ throughput with adversarial packet injection~\cite{BenderFaHe05}; for stochastic packet injections, there is much work analyzing the arrival rates under which binary exponential backoff is stable or unstable~\cite{HastadLeRo87,Aldous87,Goodman:1988:SBE:44483.44488}.\footnote{In fact, there is much work on stochastic arrival models~\cite{MetcalfeBo76,Goldberg-notes-2000,10.1145/28395.28422,10.1145/28395.28422,goldberg2004bound,shoch1980measured,abramson1977throughput,kelly1987number,aldous:ultimate,Goodman:1988:SBE:44483.44488,Goodman:1988:SBE:44483.44488,mosely1985class,DBLP:journals/siamcomp/RaghavanU98,DBLP:conf/stoc/RaghavanU95,GOLDBERG1999232}.}  Newer backoff protocols can achieve $\Theta(1)$ throughput~\cite{BenderKoKu20,
BenderFiGi19,
BenderFiGi16a,
BenderKoPe18,
BenderKoPe16,
ChangJiPe19,
awerbuch:jamming} even under adversarial packet injection.

\paragraph{Modern devices do not need exclusive access to the channel to broadcast.}  
An exciting development over the last decade has been improvements in modern radio hardware, in particular, the ability to receive useful information when more than one device transmits at the same time~\cite{CasiniGH07,GollakotaKa08,KattiRHKMC06,KattiRHKMC08,EbrahimiLK20,TehraniDN11,Andrews05,HalperinAAW07,KattiGK07}.  This information is not sufficient to immediately recover the messages sent; however, when a base station receives enough information about a set of messages, it can then decode the original messages. 
The point is that on modern hardware, backoff protocols do not need to guarantee exclusive access to the channel, because exclusive access is no longer necessary.

In fact, these capabilities have now been exploited for many years by systems that rely on interference cancellation~\cite{CasiniGH07,TehraniDN11,Andrews05,HalperinAAW07}, physical layer network coding~\cite{EbrahimiLK20,KattiRHKMC06,KattiRHKMC08,FengSK13}, and analog network coding~\cite{KattiGK07}  to improve efficiency.  For example, the celebrated paper on Zigzag decoding~\cite{GollakotaKa08} showed how a base station can decode messages sent by multiple devices simultaneously.  (In one scenario cited in their paper, this reduced packet loss from approximately 72\% to less than 1\%!) 

ZigZag decoding was built on top of 802.11, i.e., \mab{can we describe this `` industrial-strength exponential backoff''?}
exponential backoff; it showed that, in practice, improved signal-decoding techniques can yield significant improvements, even while using traditional exponential backoff for contention resolution.

\paragraph{Coding-compatible contention resolution.}
In this paper, we address the following question, taking an algorithmic approach: Can we design a backoff protocol that is better than exponential backoff for this setting where exclusive channel access is not required, that is, when devices can have simultaneous transmissions?  To do so, we define the \defn{Coded Radio Network Model} (see \secref{model}), which generalizes traditional radio network models~\cite{ChlamtacKu85} 
(where devices \emph{do} need exclusive access to the channel to transmit).  
We then ask: what is the best throughput that can be achieved?

To see an example of the throughput achieved by some standard approaches, consider the case where the base station can only decode a message if there is exactly one transmitter.  Then, if there are $n$ devices and each device broadcasts in each round with probability $1/n$, the system achieves throughput $1/e$ (and many consider this to be the ideal achievable throughput in many situations). With binary exponential backoff, if $n$ devices begin the protocol at the same time, it takes $\Theta(n \log{n})$ time for all the packets to complete, and the throughput achieved is $1/\log{n}$~\cite{BenderFaHe05}. Recent work has shown how to achieve $\Theta(1)$ throughput even in an adversarial setting~\cite{BenderFiGi19,ChenJiZh21,ChangJiPe19}, and Chang, Jin, and Pettie~\cite{ChangJiPe19} showed how to achieve throughput $1/e - O(\epsilon)$, for any $\epsilon > 0$.  

At the other end of the spectrum, several papers have proved lower bounds that establish the best achievable throughput in different situations.  For example, in a ``full sensing'' model, no protocol can achieve better than 0.568 throughput~\cite{TsybakovLi87,Goldberg-notes-2000}; in an ``acknowledgment-based'' model, no protocol can achieve better than 0.530045 throughput~\cite{GoldbergJeKa00,GoldbergJeKa04,Goldberg-notes-2000}.

This raises the question: \emph{what is the best achievable throughput in a system that uses modern techniques to decode simultaneous signals?}
If we can send an arbitrary number of messages simultaneously, it may seem as if there is no inherent limit.  This is not the case.  There are two main reasons, one pragmatic and one information theoretical.  First, from a practical perspective, there is a limit to the number of simultaneous signals that a base station can usefully interpret.  For each signal received, there is noise, fading, frequency offset, inter-symbol interference, etc., all of which means that the more simultaneous signals, the less useful information the base station receives (see, e.g.,~\cite{TseV05}, Chapter 3). Thus, there is some \defn{decoding threshold} $\kappa$ where if there are more than $\kappa$ simultaneous transmissions, the base station learns nothing useful. A key implication here is that there is a maximum achievable throughput, and contention resolution is still important---even though we no longer need to 
guarantee exclusive channel access for the transmitter.

From the perspective of information theory, there is an even stronger limitation: none of the techniques we are considering (e.g., iterative interference cancellation) has increased the actual capacity of the communication channel.\footnote{Typically, the capacity of a wireless channel depends on other factors, e.g., the transmission frequency, the modulation scheme, etc.  As we do not want to address those issues in this paper, we simply define one time slot as the time it takes to send a single message using the system at hand.} It simply provides a method for using that capacity more efficiently. For example, if two messages are being transmitted simultaneously, it will then take twice as long to receive and decode them. This is observed in practice in ZigZag Decoding~\cite{GollakotaKa08}, where the authors state: ``when senders collide, ZigZag attains the same throughput as if the colliding packets were a priori scheduled in separate time slots.''  The conclusion, then, is that the maximum achievable throughput is still $1$.

\paragraph{Results.}
In this paper, we develop a randomized backoff protocol designed for a system that can support multiple simultaneous transmissions. Our model captures the fact that the base station can only receive useful information if there are $\leq \kappa$ simultaneous transmissions, and it reflects the fact that (conceptually) if $\ell$ messages are transmitted simultaneously, the base station only receives roughly $1/\ell$ of each original message. The new algorithm achieves optimal throughput arbitrarily close to $1$, approaching $1$ as $\kappa$ grows. More precisely, it achieves a throughput of $1 - \Theta(1/\log{\kappa})$.  (And in doing so, it shows that the increased power of simultaneous transmission is non-trivial, as it breaks the constant-throughput lower bounds~\cite{TsybakovLi87,GoldbergJeKa00,GoldbergJeKa04,Goldberg-notes-2000}.)

The basic idea of the protocol is relatively simple: it uses a backon-backoff mechanism to construct groups of messages that will broadcast simultaneously until they are decoded. The backoff-backon mechanism is tuned to yield groups of size at most (the decoding threshold) $\kappa$, and a simple ``admission control'' mechanism prevents newly arrived packets from disrupting the ongoing protocol. (Ensuring that the same packets broadcast together repeatedly is also critical to protocols like ZigZag Decoding~\cite{GollakotaKa08}---it is the repeated overlaps that make simultaneous decoding easier.)
Our analysis relies on a potential-function argument, which shows that, at any given time, we are either decoding groups of packets or moving closer to a good group structure.

The key challenge  is getting throughput $1 - o(1)$.   Even wasting a few rounds to perform backoff may render our goals impossible to achieve.  For example, if there are $x$ packets injected into the system, and we spend some $\Theta(x)$ rounds running backoff to construct batches of size $\kappa$, then our throughput will not be close to $1$.  Indeed, a classical backoff protocol like exponential backoff~\cite{MetcalfeBo76} spends at least $\Theta(x)$ time even just growing its window to size $\Theta(x)$.  Thus, in both the protocol design and the analysis, we have to be particularly careful to prevent constant factors from creeping in!

To show good throughput, we focus on the \emph{rate} at which packets can be added to the system without creating a large backlog of undelivered packets. If packets arrive at a rate equal to the channel capacity, and yet there is never a backlog, then we conclude that the backoff protocol is delivering a throughput equal to the channel capacity.  (We are not interested in executions where there are very few packets---the throughput may look low simply because there is nothing to do.)  We allow for arbitrary adversarial packet arrival, i.e., packet arrivals can be bursty, smooth, or whatever.  (There is no assumption of random packet arrivals.)  

\mab{Do we want to mention that we have 
 models from adversarial queueing theory (AQT)~\cite{borodin2001adversarial,cruz-one,cruz-two,andrews2001universal}?}
 
To show that we achieve throughput arbitrarily close to $1$, we assume that for some window size $w$, there are at most $(1 - o(1))$ packet arrivals in every window of size at least $w$, i.e., at some level of granularity the arrivals do not exceed the channel capacity.  We show that in this case, there is never \john{whp} a large backlog of packets in the system, i.e., the system continues to deliver packets. In addition, we show that there is no starvation, i.e., that packets do not stagnate in the system too long and will eventually be delivered.
\begin{theorem}
Suppose that in every window of size $w$, for some $w \geq 16\threshold^2$, there are at most $\big(1 - \frac{5}{\ln \threshold}\big)w$ packet arrivals. Then the \sysname 
guarantees the following two conditions in the Coded Radio Network Model:\\
(i) At any given time, the number of packets in the system is at most~$2w$, with probability $\geq 1 - 1/\poly{w}$.\\
(ii) A given packet is delivered within $O(\sqrt{\kappa} w \ln^3{w})$ time slots, with probability $\geq 1 - 1 / \poly{w}$.
\label{thm:main}
\end{theorem}
\noindent These results are proved in \thmref{throughput} and \thmref{latency}.

One thing to note is that the window size $w$ is not known to the algorithm---it is simply an analytic tool. Thus the execution can have entirely arbitrary packet arrivals and the theorem will hold if there exists some $w$ for which the arrival rate is satisfied. Another note is that, despite the high probability bounds, the execution is not limited to polynomial lengths; the claims will hold for any time $t$, even in a very long execution.  Finally, observe that for packet latency, the best possible result is $O(w)$: if the adversary injects $w$ packets in round $1$, it will necessarily take at least $w$ rounds for all of them to complete. (In fact, we consider the batch case in \thmref{batch:latency}.)

\paragraph{Other related work.}  %
There is a long history of algorithmic research analyzing the performance of backoff protocols, looking at different aspects. As in this paper, the primary goal has been to analyze packet delivery rates, whether in the form of throughput, total time, or stability~\cite{ChangJiPe19,
MetcalfeBo76,
BenderFaHe05,
bender:heterogeneous,
hastad:analysis,
GoodmanGrMaMa88,
RaghavanUp99,
GOLDBERG1999232,
hastad:analysis,
goldberg:contention,
bianchi:performance,
song:stability,
DEMARCO20171,
DeMarcoSt17}. %
An alternative metric has been looking at how long it takes to deliver \emph{one} message~\cite{jurdzinski2015cost,
jurdzinski2005probabilistic,
jurdzinski2002probabilistic,
jurdzinski2002probabilistic,
farach2007initializing,farach2015initializing,
Chlebus:2016:SWM:2882263.2882514,
ChrobakGK07,
ChlebusGKR05,
DBLP:journals/dam/MarcoPS07,
willard:loglog,fineman:contention,
fineman:contention2}. 

Backoff has been studied in a wide variety of different models.  Many of the above papers have assumed that radios can detect collisions; some work has focused on the case where there is no collision detection~\cite{BenderKoKu20,ChenJiZh21,DeMarcoKoSt21,DeMarcoSt17}. The large majority of these protocols have been randomized, but there has been some consideration of deterministic protocols~\cite{doi:10.1137/140982763,
DeMarcoSt17,
DBLP:conf/icdcs/MarcoKS19,
Kowalski05}.  A few papers have focused on more realistic wireless models that capture signal-to-noise ratio~\cite{FinemanGKN16,FinemanGKN19}. There has been some significant recent work on designing backoff protocols that are robust to noise and/or jamming~\cite{awerbuch:jamming,richa:jamming2,
richa:competitive-j, richa:jamming4,
richa:jamming3,
DBLP:journals/ton/RichaSSZ13,
ogierman:competitive,
ogierman2018sade,
DBLP:journals/tcs/AntaGKZ17,
anantharamu2019packet,anantharamu2011medium,
Chlebus:2016:SWM:2882263.2882514,
BenderFiGi19}.
A critical modelling question is how packets are injected, e.g., in batches~\cite{GoldbergJeLeRa97, BenderFaHe05,BenderFiGi06,
  Gereb-GrausTsa92,GreenbergFlLa87,DBLP:journals/jacm/GreenbergW85,willard:loglog,FernandezAntaMoMu13}, stochastically~\cite{MetcalfeBo76,10.1145/28395.28422,10.1145/28395.28422,goldberg2004bound,shoch1980measured,abramson1977throughput,kelly1987number,aldous:ultimate,Goodman:1988:SBE:44483.44488,Goodman:1988:SBE:44483.44488,mosely1985class,DBLP:journals/siamcomp/RaghavanU98,DBLP:conf/stoc/RaghavanU95,GOLDBERG1999232}, or via an adversarial queuing theory model~\cite{BenderFaHe05,
  GreenbergFlLa87,willard:loglog,GoldbergJeLeRa97,GoldbergMaRa99,BenderFiGi06,FernandezAntaMoMu13,ChlebusKoRo12,ChlebusKoRo06,anantharamu:adversarial-opodis}.
There has also been a variety of work looking at variations, e.g., the packets may have deadlines~\cite{AgrawalBeFi20}, or the packets may have different sizes~\cite{BenderFiGi06}. 

Our algorithm is built on the multiplicative-weight~\cite{AroraHK12} approach, which has long been used for contention resolution, e.g., in TCP congestion control~\cite{ChiuJ89}.  More recently, Awerbuch, Richa, and Scheideler~\cite{awerbuch:jamming} 
used a multiplicative-weight approach to solve backoff in the presence of jamming and noise.  And Chang, Jin, and Pettie~\cite{ChangJiPe19} used multiplicative-weights updates to achieve $1/e - o(1)$ throughput.



\section{Coded Radio Network Model}
\label{sec:model}

The goal of this section is to develop the Coded Radio Network Model.
We begin with background on how radios transmit and receive signals, before continuing to present the model itself. 

\subsection*{Physical-layer radios}

Here we briefly describe how physical-layer radios work.  (See, e.g.,~\cite{TseV05}, for more details.)
The packet to be sent is broken up into smaller pieces. Each of these pieces is then ``modulated,'' i.e., mapped to a complex number. The mapping depends on the modulation scheme, which impacts the transmission rate.  For example, a simple binary encoding scheme might map a $1$ to $e^{0i} = (1,0)$ and $0$ to $e^{\pi i} = (-1, 0)$.  These complex numbers correspond to waveforms which the radio transmits.  (For example, in ``phase-shift keying,'' the complex number determines the \emph{phase} of the wave-form; for the binary encoding above, the phase is 0 to transmit a 1 and $\pi$ to transmit a 0.)

The receiver can observe the waveforms, which have been partially corrupted during the transmission: the signals are attenuated (i.e., the amplitude reduced), the phase has shifted based on the distance between the transmitter and receiver, there is additional noise, etc.  Despite these challenges, modern receivers are designed to synchronize with the transmitter, estimate the various error parameters, and transform the waveform back into a sufficiently close approximation of the original complex number.  These are then demodulated, and the packet reconstructed.

A notable aspect is that the process is \emph{additive}: if two complex number are transmitted simultaneously, the receiver is able to decode the \emph{sum} of the original complex numbers (each partially corrupted, as before, and with noise added).  This sum can provide useful information about the original messages; for example, if you already know one of the messages, you can subtract it from the sum and recover the second message.  This is famously exploited in systems that rely on interference cancellation, e.g.,~\cite{KattiRHKMC06,KattiRHKMC08,Andrews05,HalperinAAW07}.

Too many simultaneous transmissions, however, will make this demodulation increasingly difficult because of the noise and signal corruption that is added while the messages are in flight.

\subsection*{Our model}

In this section, we describe the Coded Radio Network Model, which models radios that can receive multiple packets transmitted simultaneously, while respecting the information-theoretic capacity of the channel and the decoding limits inherent to modern hardware.
\mab{Ok to use  ``message'' and ``packet'' or switch just to ``packet''?}

We assume that time is divided into (synchronized) slots, each of which is sufficiently large to send a single packet.
At any given time, there are a collection of packets to be sent in the system. 
(Our discussion is going to focus on the \emph{packets}, not the devices sending them, which may lead to some strange elocutions in which a message may seem to have agency in transmitting itself.) 

There is a base station that is intended to receive all the packets. In a more traditional contention-resolution model, the base station would receive a packet that was broadcast by itself (i.e., the packet broadcast itself) in a time slot. Here, when there are simultaneous transmissions, we will allow the base station to recover some information (i.e., at a lower level in the stack, it recovers sum information).

We define $\kappa$ as the \defn{decoding threshold} for the system, i.e., the threshold determined by the hardware for how many simultaneous transmissions can yield useful information.
We say that a time slot is \defn{good} if there is at least one and at most $\kappa$ transmissions in that time slot.
We say that a time slot is \defn{bad} if there are more than $\kappa$ transmissions in that  slot. 
If no packets broadcast in the slot, it is a \defn{silent} slot. (Silent slots are neither good nor bad.)
The packets in the system cannot tell  which steps are good and bad. They can hear which slots are silent, however.
\mab{Changed to say that a silent step is neither good nor bad. Ok? Is the introduction of silent slots here ok?}

The execution will consist of \defn{decoding windows}, i.e., periods of time during which the base station has collected enough information to decode packets.
Packets that have been decoded are delivered and leave the system.
\mab{May I add leave the system sentence?}
Since a base station only learns useful information during a good time slot, it can only decode packets sent during good time slots---and ignores packets sent only during bad time slots.
We focus on decoding windows where \emph{every} packet that is broadcast during a good time slot must be decoded.  

A decoding window needs to be long enough that the base station can decode all the packets sent during good slots. The information-theoretic constraint implies that if we are to decode $j$ packets, we are going to need at least $j$ good time slots. (In the next section, we explain the linear-algebraic intuition why $j$ time slots is sufficient.)

At the end of a successful decoding window, a \defn{decoding event} occurs:
\begin{definition}
\label{def:decodingevent}
There is a \defn{decoding event} of size $j$ at time $t$ if there exists a
time window that begins with a good slot and ends at time $t$, such that:
\begin{enumerate}[noitemsep]
\item The window contains no decoding events.
\item Exactly $j$ packets broadcast in the window in good slots.
\item The number of good slots in the window is at least $j$.
\end{enumerate}
That window is the \defn{decoding window}.
\begin{enumerate}[noitemsep]
\setcounter{enumi}{3}
\item Decoding events occur iteratively over time, meaning a decoding event occurs the first time that these conditions are satisfied after the previous decoding event.
\end{enumerate}
When a decoding event of size $j$ occurs, the $j$ packets that broadcast in good slots during the decoding window are received by the base station.
\end{definition}

Every device knows when a successful decoding event occurs.

\paragraph{Examples.}
If three packets broadcast simultaneously for three rounds (and $3 \leq \kappa$), then those three rounds yield a decoding event of size $3$.  Alternatively, if all three packets $(a,b,c)$ broadcast in time slot $1$ simultaneously, while packets $b$ and $c$ broadcast in time slot $2$ and packet $c$ broadcasts alone in time slot $3$, then after these three time slots a decoding event occurs and all three packets are delivered.  Notice that this definition satisfies our intuition that it should always take at least $j$ time slots to successfully deliver $j$ packets.

This formalization of a decoding event does not take advantage of all of the information that the decoding technology may be able to obtain. For example, consider the schedule in which packets $a$ and $b$ broadcast together in time steps $1$ and $3$. In step $2$ a single packet $c$ broadcasts, i.e., there is a decoding event in the window comprised of only step~$2$.  According to \defref{decodingevent}, we cannot have a decoding event for steps~1-3, because there was already a decoding event at step~2. Since decoding windows are required to be disjoint, the broadcasts from step $1$ cannot be used for any future decoding window. Even though the broadcasts in steps $1$ and $3$ provide sufficient information to decode packets $A$ and $B$, that information is lost.

In fact, this paper will only need (and have) decoding intervals of size $O(\kappa)$.  In practice, decoding windows will also be limited in length, e.g., the base station should not have to remember and try to decode arbitrarily long sequences.

\subsection*{Practicalities}

Here, we are going to give some intuition for how random linear network coding might be used to implement the model described above. (Complete physical-layer implementations are beyond the scope of this paper.)

From the perspective of linear network coding, each time slot is a column vector $t$ with one entry for each possible transmitter: $t[j] = 1$ if device $j$ broadcasts and $t[j] = 0$ if device $j$ is silent.
Over the collection of good time slots that form a decoding window,  the collection of vectors forms a binary matrix $T$ containing one column for each good slot. (We drop the empty and bad time slots.)  

If device $j$ is transmitting message $m[j]$, then due to the additive nature of radio transmission, the base station receives the (row) vector $mT$ during the decoding window.  If the matrix $T$ is invertible, then the base station can now decode the packets sent during this decoding window by computing $mTT^{-1}$.

There are many techniques for ensuring that the matrices are invertible. For example, random linear coding multiplies each message by a random coefficient before it is sent.   The resulting transmission matrix $T$ is now a random matrix and hence likely to be invertible.  An alternative is to randomize the slots in which packets broadcast to ensure that each column vector is  unique (which is all that is needed for linear independence).

A key technical issue is how the base station learns the transmission matrix.  That is, the base station needs to know which packets are being transmitted in which slots in order to decode the packets. And in the case of random linear network coding, the base station also needs to know the coefficients. Thus, there is a small amount of ``control information'' that the base station needs to know before the packets can be decoded.
 
This problem is easier to solve in the context of this paper, where we ensure that the same packets always broadcast together during each slot of a decoding window.  (For example, consider the situation where packets $(a,b,c)$ are broadcast together three times until a decoding event occurs.)  
This decision makes for both a simpler algorithm, but also one that is more compatible with existing physical layer systems.
Indeed, the requisite control information can even be sent using the ZigZag Decoding techniques~\cite{GollakotaKa08}.
In the general case, other  techniques are needed to communicate this control information, including error-correcting codes, lower communication rates, leveraging random offsets (as in ZigZag Decoding), etc.

While our discussion here has been brief, developing systems with these types of network coding capabilities has been and remains an active and exciting area of research; see e.g.,~\cite{EbrahimiLK20,KattiRHKMC06,KattiRHKMC08,FengSK13,KattiGK07,CasiniGH07}.



\section{\sysname}
\label{sec:algorithm}

In this section, we describe the \sysname, a protocol for achieving efficient channel utilization in the Coded Radio Network model.

\paragraph{Admission control.} Our algorithm separates packets into two types: \defn{active} and  \defn{inactive}. 
Newly joined packets start as inactive.
Inactive packets listen but do not broadcast; active packets do both. 

Silent steps trigger \defn{packet activations}.
When an inactive packet hears a silent slot, it becomes active.
An active packet  stays active until it successfully broadcasts and leaves the system; 
it never returns to being inactive.

\paragraph{Types of epochs.}
The algorithm divides time into \defn{epochs}, where the same set of packets broadcast in each slot of the epoch.   
The length of an epoch can vary from $1$ to $\threshold$ slots (as specified below).

Each packet $j$ has an associated \defn{joining probability $p_j$}, which means that when  a new epoch begins,
packet $j$ \defn{joins the epoch} with probability $p_j$ and then broadcasts (deterministically) in every slot of the epoch. A packet has positive joining probability if and only if it is active. 
(An inactive packet $j$ has joining probability $p_j = 0$.)  
If a packet does not join an epoch, then it remains silent until the epoch is complete.

There are three triggers that causes an epoch to end, leading to three types of epochs:
%
%
\begin{enumerate}[noitemsep]
\item \textbf{A silent slot}---This results in a \defn{silent epoch}. A silent slot implies that no packets joined the epoch.
\item \textbf{A decoding event}---This results in a \defn{successful epoch} where the packets that joined the epoch are decoded.
\item \textbf{\boldmath $\kappa$ time slots have elapsed without a decoding event}---This results in a \defn{overfull epoch}, as it implies that too many packets joined the epoch and these slots will not be part of a decoding window.
\end{enumerate}

\paragraph{\boldmath Epoch-joining probabilities.}
When a packet $j$ first becomes active, its initial joining probability is
\begin{equation}
\label{eq:initialpj}
p_j=1/\sqrt{\kappa}.    
\end{equation}

When an epoch ends, packet $j$ updates  $p_j$ as follows.    

\begin{equation}
\displaystyle
p_j \leftarrow
  \left\{ \begin{array}{ll}
\displaystyle
p_j \cdot \threshold^{1/4} &   \mbox{upon a silent epoch,} 
\vspace{2ex}\\
\displaystyle
\frac{p_j}{\threshold^{1/4}} &   \mbox{upon an overfull epoch,} 
\vspace{2ex}\\
\displaystyle
p_j  &   \mbox{upon a successful epoch.}
          \end{array} \right.
 \label{eq:changepj}
\end{equation}

\paragraph{Algorithm highlights.}

The algorithm itself is concise, but delicate: 
at the end of an epoch, packets multiplicatively update their broadcast probabilities by a factor of $\kappa^{1/4}$. Packets join with broadcast probability $p_0 = 1/\sqrt{\kappa}$. Inactive packets only become active upon hearing a silent epoch. Our target contention is $c_* = \sqrt{\kappa}$. If contention is in the window $[\kappa^{1/4}, \kappa^{3/4}]$, we say that \defn{contention is good}.

There are three important aspects that are critical to its success. 
First, the epoch structure enables coordination among the packets so that packets retransmit together, which enables them to be decoded.\footnote{It is also the case that simultaneous retransmissions simplify the lower physical-layer implementation, as discussed in \secref{model}. 
For example, ZigZag Decoding~\cite{GollakotaKa08} similarly relies on simultaneous retransmissions for this reason.}  
(Contrast this structure with a protocol in which devices broadcast independently in every round; it is much less likely to get enough information on a collection of packets that can be simultaneously decoded.)
Second, careful admission control is important because it ensures that newly arriving packets do not interfere with packets that are currently participating in epochs.

Third, and most critically, setting the $p_j$'s is  delicate, given the objective of $1-o(1)$ throughput.
It is critical that newly active packets start with a $o(1)$ joining probability; see \eqref{initialpj}, as otherwise too much time might be spent backing off. (Recall that even a small number of overfull epochs is too many for $1-o(1)$ throughput.)
Moreover, it is critical that the $p_j$s change rapidly.
To put \eqref{changepj} in context, the ``contention'' on the channel  (expected number of broadcasts in a slot) changes much  faster (exponentially) than with exponential backoff~\cite{MetcalfeBo76}.
Contention even changes  a $\omega(1)$-factor faster than the
multiplicative-weight-update algorithm of 
Chang et al.~\cite{ChangJiPe19}.



\section{Potential Function}
\label{sec:potentialfunction}

Throughout the execution, we  maintain a potential function $\Phi(t)$ that captures the state of the system at $t$ and measures the progress toward delivering all packets.
We will see that packet arrivals increase $\Phi(t)$ by $1 + o(1)$ per newly arrived packet (\lemref{arrivals}), that packets exiting the system successfully decrease $\Phi(t)$ by $1$ per packet (\lemref{success}), and that an epoch of length $\ell$ is either a successful epoch or causes contention to move closer to the desired contention of $c_*$ with high probability, and  will decrease $\Phi(t)$ by $\ell (1 - o(1))$ in those cases (\lemref{error_unlikely}, \lemref{non-error_good}).

Recall that $p_i$ is the joining probability for each packet $i$.  
We define the \defn{contention} $\boldsymbol{c_t} = \sum_i p_i$.  We define $\boldsymbol{\pmin(t)} = \min_{\{i : i \text{ is an active packet}\}} p_i$
to be the minimum probability of any \emph{active} packet currently in the system. If there are no active packets in the system, we define $\pmin(t) = 1$. We define the \defn{target contention} $\boldsymbol{c_*} = \sqrt{\threshold}$. $N_t$ is the total number of packets in the system at time $t$, including both active and inactive packets. $M_t$ counts the number of inactive packets in the system at time $t$.

The potential function consists of four components (named below).

\[
\Phi(t) = N_t 
\,+\, 
\max\left\{0,\,  4 \threshold \log_{\threshold} \Big(\, \frac{c_t}{c_*}\,\Big)\right\} 
\,+\, 
4 \log_{\threshold} \! \Big(\, \frac{1}{\pmin(t)}\, \Big) 
\,+\,
\frac{5\, M_t}{\ln \threshold}\,.
\]

\begin{enumerate}
    \item $\boldsymbol{N(t)} = N_t$. Since no term can go negative, this term implies that if the potential is small, there cannot be many packets remaining in the system.
    
    \item $\boldsymbol{\mbox{\bf logC}(t)} = \max\big\{0,\, 4\threshold \log_{\threshold} 
    \big(\frac{c_t}{c_*}\big)\big\}$. When  positive, the term indicates how far the contention is from ideal, i.e., how many epochs of length $\threshold$ are needed to reach $c_*$.

    \item $\boldsymbol{s(t)} = 4\log_\threshold\big(\frac{1}{\pmin(t)}\big)$.
    The term indicates how many silent epochs (of length $1$) would be necessary to bring the packet with minimal join-probability up to join probability $1$. 

    \item $\boldsymbol{u(t)} = \frac{5 M_t}{\ln \threshold}$.
    The term compensates  for packet activations. Potential is stored in this function for each arrived, inactive packet. That potential then compensates for the increase in the second term caused by the rise in contention when the packets activate.
\end{enumerate}

It is immediate that the potential function is always positive when there are packets in the system. 
Since all the terms are non-negative, it also is immediately true that $N(t) \leq \Phi(t)$, i.e., the potential function is an upper bound on the number of packets in the system. If we can show that the potential in the system remains low, then we can conclude that the number of packets in the system also remains low.

It is constructive to compare our potential function with the function from Chang et al.~\cite{ChangJiPe19}.
Similarly, to Chang et al., $\Phi(t)$ has  terms that measure the number of packets, how far contention is from the algorithm's target contention, and the smallest broadcast probability.
Our design of 
$\Phi(t)$ also takes  into account that 
(1) epochs  can be different lengths, (2) there are transitions from inactive to active packets, (3) our multiplicative-weight update rule is especially aggressive, and (4) only a  $o(1)$ fraction of slots can be bad.

\section{Throughput Analysis}
\label{sec:channelutilization}

In this section, we will show that the \sysname guarantees $1 - o(1)$ throughput. Most of the section is devoted to analyzing the potential function, specifically showing that it decreases by $1 - o(1)$ in expectation in every round where it is $> 6\threshold$, and increases by $1 + o(1)$ for every packet that is injected by the adversary). Thus, in \thmref{throughput}, we can conclude that as long as no more than $w(1 - 5/{\ln \threshold})$ packets arrive in a window of length $w$, then (with high probability) there are never more than $2w$ packets in the system. 

We first observe that successful epochs correspond to decoding windows:
\begin{lemma}
\label{lem:decodingwindow}
An epoch is successful if and only if $\leq \threshold$ packets join the epoch.  In that case, there is a decoding event at the end of the epoch and the decoding window corresponds exactly with the epoch.
\end{lemma}
\begin{proof}
If $\leq \threshold$ packets join the epoch, then those packets will broadcast together in each of the following $\leq \threshold$ rounds of the epoch.  These broadcasts cause a decoding event at the end of the epoch which includes those packets.  

The last slot in a decoding window must be a good slot. (If the last slot were bad, then the conditions would have been satisfied earlier and the decoding event would have occurred earlier.)  If $> \threshold$ packets join an epoch, then every round in the epoch is bad, and hence no decoding event occurs at the end of the epoch.  If no packets join an epoch, then again the last slot of the epoch is not good and so again, no decoding event occurs.

The decoding window cannot be a suffix of the successful epoch that it concludes for the following reason: every slot in a successful epoch is good (because $\leq \threshold$ packets joined the epoch); if the decoding event delivers $j$ packets, then it must occur after exactly $j$ slots of the epoch---no earlier, and no later. If there were even one slot in the epoch that preceded the decoding window, then, the entire decoding window could have moved one slot earlier.

The decoding window cannot extend to the left of the successful epoch that ends with the decoding event.  That is because between the last decoding event and the beginning of this epoch, every slot is either bad or empty, as each of the intervening epochs is either overfull or silent.  And, as per the definition of the decoding window, it must begin with a good slot.

Thus we conclude that the decoding window corresponds exactly with a successful epoch which $\leq \threshold$ packets joined.
\end{proof}

\subsection{Analyzing the probability of an error epoch}

\begin{definition}
We define an \defn{error epoch} to be an epoch that is either a silent epoch that occurs when contention $c_t \geq \threshold ^{1/4}$ or an overfull epoch that occurs when contention $c_t \leq \threshold ^{3/4}$. Note that this is well-defined since contention is constant throughout an epoch.
\end{definition}
\seth{Contrast error epoch to successful/overfull epochs.}

\begin{lemma}
An epoch beginning at time $t$ is an error epoch with probability at most $1 / 2^{\Theta(\threshold^{1/4})}$.
\label{lem:error_unlikely}
\end{lemma}
\begin{proof}
An error epoch is a silent epoch that occurs when $c_t \geq \threshold^{1/4}$ or an overfull epoch that occurs when $c_t \leq \threshold^{3/4}$.

Applying Chernoff bounds, the probability of an epoch being silent (i.e. no packets joining the epoch) when $c_t > \threshold^{1/4}$ is at most $1 / 2^{\Theta(\threshold^{1/4})}$. Similarly, the probability that an epoch is overfull (i.e. more than $\threshold$ packets join) when $c_t < \threshold^{3/4}$ is at most $1 / \threshold^{1/4 (\Theta(\threshold))} = 1/\threshold^{\Theta(\threshold)}$. Applying a union bound, the probability of an error epoch occurring is at most $1/2^{\Theta(\threshold^{1/4})}$. \john{Need some constants.}
\end{proof}

We now bound the number of error epochs during an interval.
\begin{lemma}
\label{lem:intervalerror}
Let $t_2$ be a time slot.  Then for all intervals $I = (t_1, t_2]$ where $t_1 < t_2$, of size $w = t_t - t_1$ satisfying $\sqrt{w} > \Theta(\threshold^{1/4})$, the number of 
error epochs occurring in $I$ is at most $\sqrt{w} + c(t_2 - t_1) / 2^{\Theta(\threshold^{1/4})}$, for any fixed constant $c \geq 6$, 
with probability $\geq 1 - 1/ \poly{w}$.
\end{lemma}
\begin{proof}
Choose an arbitrary slot $t_1 < t_2$, and define interval $I = (t_1, t_2]$.  Let $j$ be the random variable denoting the number of error epochs that occur in $I$. By Lemma~\ref{lem:error_unlikely}, $\mathbb{E}[j] \leq (t_2 - t_1) / 2^{\Theta(\threshold^{1/4})}$, as $I$ contains at most $t_2 - t_1$ distinct epochs. Let $X_i$ be an indicator random variable that is $1$ iff the $i$th epoch in
$I$ is an error epoch. Note that $\Pr[X_i = 1] \leq 1 / 2^{\Theta(\threshold^{1/4})}$, 
even if we condition on the values of $X_j$ for all $j \neq i$. 
Thus, we can apply Chernoff Bounds to find:
\[\Pr\left[j \geq \sqrt{w} + c(t_2 - t_1)/2^{\Theta(\threshold^{1/4})}\right] \,\leq\, 1 / 2^{\sqrt{w} + c(t_2 - t_1) / 2^{\Theta(\threshold^{1/4})}}\,.\]
Union bounding over all $t_1$ from $t_1 = 0$ to $t_1 = t_2$, we find that the probability that the number of error epochs in $I$ exceeds $\sqrt{w} + c(t_2 - t_1) / 2^{\Theta(\threshold^{1/4})}$ for any interval $I = (t_1, t_2]$, over all $t_1$, is at most

\[\sum_{t_1 = 0}^{t_2}  1 / 2^{\sqrt{w} + c(t_2 - t_1) / 2^{\Theta(\threshold^{1/4})}} \,\leq\, (2/c) 2^{\Theta(\threshold^{1/4}) - \sqrt{w}} \,. 
\]

The inequality follows due to the sum being over a geometric series.
\end{proof}

\subsection{Evolution of the Potential Function}

\mab{John, note that busy $\rightarrow $ overfull.}
Consider an epoch of length $\ell$ during which $i$ packets are inserted by the adversary, for any arbitrary $(\ell, i)$. We prove in this subsection that error epochs are the only epochs in which the potential $\Phi(t)$ fails to decrease by at least $\ell(1 - 1/\threshold) - i(1 + 5/\ln \threshold)$ (\lemref{non-error_good}). 

First, we analyze the impact of packets joining the system on the potential function.  This takes place in two steps.  First, they arrive and are inactive.  This only increases the first and last terms of the potential function:

\begin{lemma}
Each packet that arrives during in slot $t$ increases the potential by $\Phi(t)$ by $1 + 5/\ln \threshold$.
\label{lem:arrivals}
\end{lemma}
\begin{proof}
Newly inserted packets arrive as inactive packets, and thus they have no impact on contention or $\pmin$. The only terms they affect are the first and last, which increase by $1$ and $5/\ln \threshold$, respectively, per newly inserted packet.
\end{proof}

Next, we consider what happens to the potential function when packets activate.  This activation occurs only during a silent epoch, and can  impact only the second, third, and fourth terms of the potential function:
 
\begin{lemma}
If $\Phi(t) > 6 \threshold$ and slot $t$ is part of a non-error epoch, then packet activations do not increase $\Phi(t)$.
%
More precisely, packet activations in slot $t$ increase $\Phi(t)$ only in the following situations:
\begin{itemize}[noitemsep]
    \item $\Phi(t) < 6 \threshold$ and $\pmin(t) > 1/\sqrt{\threshold}$.
    \item Slot $t$ is part of an error epoch, and $\pmin(t) > 1/\sqrt{\threshold}$.
\end{itemize}
In these cases, $\Phi(t)$ increases by $< 2$ because of packet activations during the (silent) epoch.
\label{lem:activation}
\end{lemma}
\begin{proof}
Upon hearing a silent slot, all inactive packets activate.  Assume that $r = M_t$ packets activate.
Notice that if $\pmin > 1/\sqrt{\threshold}$, then $\pmin$ decreases when packets activate.

We first show that if $\pmin$ does not change due to the packet activations, then the potential does not increase due to packet activations.  To show this, assume that $\pmin$ does not change when the packets are activated.  Consider the four terms of the potential function:
\begin{itemize}
    \item $N(t)$ does not change based on packet activations (as it is only affected by new packet injections).
    \item The second term of the potential function $\mbox{logC}(t)$ increases by at most $4r/\ln{\threshold}$ because: If $c_t < c_*$, then $\mbox{logC}(t) = 0$; hence we only need to consider the increase in contention above $c_* = \sqrt{\threshold}$.  That is, the maximum increase in $\mbox{logC}(t)$ is at most 
    $4\kappa\log_{\kappa}(1 + (r/\sqrt{\kappa})/\sqrt{\kappa}) \leq 4\kappa r / (\kappa\ln{\kappa}) = 4r/\ln{\threshold}$ (this inequality follows using $\ln(1+x) \leq x$). 
    \seth{John, is there an easy way to explain the (obvious) fact that $c_t = c_*$ is the worst-case?}
    \item $s(t)$ does not increase, by assumption for this case.
    \item $u(t)$ decreases by $5r/\ln{\threshold}$ because $r$ inactive packets became active.
\end{itemize}
We conclude that if $\pmin$ is not changed by the packet activations (and hence $s(t)$ is not increased), then the increase in potential due to $\mbox{logC}(t)$ is less than the decrease in potential due to $u(t)$, and hence $\Phi(t)$ does not increase.  

Next, we show that even if $\pmin$ does change due to packet activations, the potential never increases by more than $2$. Consider the third term of the potential function $s(t)$.  If $\pmin$ decreases, then $s(t)$ increases. A packet, when activated, begins with $p_j = 1/\sqrt{\kappa}$.  Thus, the most that $s(t)$ can increase is $< 4\log_{\threshold}{\sqrt{\threshold}} = 2$ since $\pmin < 1$ during a silent slot.  In total, summing the changes over all four parts of the potential function, the potential increases by at most: $4r/\ln{\threshold} + 2 - 5r/\ln{\threshold} < 2$.

It remains to show that $\Phi(t)$ increases (by at most 2, as already shown) only when $\pmin$ is reduced due to packet activations \emph{and} either $\Phi(t) < 6\threshold$ or there is an error epoch.  Consider the following three cases (immediately prior to the packet activations):

\paragraph{Case 1: The number of active packets \boldmath$< \threshold$ and $\pmin$ is reduced by packet activations.}  If $r/\ln \threshold \geq 2$, then there is no increase in potential due to packet activations because $(4r/\ln{\kappa} + 2) - 5r/\ln{\kappa} \leq 0$. Otherwise, $r < 2\ln{\threshold}$.

In this latter case (where the potential might increase and $r < 2\ln{\threshold}$), we can compute the potential function before the packet activations as follows: 
\begin{itemize}
    \item $N(t) \leq \threshold + r$, by the case assumption.  Since $r \leq 2\ln{\threshold} \leq \threshold$, we conclude that $N(t) \leq 2 \threshold$.
    \item $\mbox{logC}(t) \leq 4\threshold \log_{\threshold} (c_t/\sqrt{\threshold}) \leq 2\threshold$ because $c_t \leq \kappa$.
    \item $s(t) \leq 4\log_{\threshold}(\sqrt{\threshold}) \leq 2$; because $\pmin$ is reduced by the packet activations, we know that before these activations, it must have been the case that $\pmin \geq 1/\threshold$.
    \item $u(t) = 5r / \ln{\threshold} < 10$.
\end{itemize}
Thus, $\Phi(t) \leq 2\threshold + 2\threshold + 2 + 10 \leq 4\threshold + 12$.  Given that $\threshold \geq 6$, we conclude that $\Phi(t) \leq 6\threshold$.  Thus, in this case, the potential increases only if $\Phi(t) \leq 6\threshold$.

\paragraph{Case 2: The number of active packets \boldmath$\geq \threshold$ and $\pmin$ is reduced by packet activations.}
Since $\pmin$ is reduced by the packet activations, we know that prior to their activation, $\pmin > 1/\sqrt{\threshold}$.  It follows that $c_t > \threshold/\sqrt{\threshold} = \sqrt{\threshold}$. We know that this is a silent epoch, because otherwise the packets would not have activated.  And the contention at the beginning of this silent epoch was $> \sqrt{\threshold}$.  Thus, this epoch is an error epoch, by definition.

\paragraph{Case 3: \boldmath $\pmin$ is \emph{not} reduced by packet activations.} As we have already shown, the potential does not increase.

Over all three cases, then, we see that potential only increases when $\Phi(t) \leq 6\threshold$ and $\pmin$ is reduced by the packet activations, or when there is an error epoch and $\pmin$ is reduced by the packet activations.  And in all such cases, the potential increases by at most 2.
\end{proof}

\begin{lemma}
A successful epoch of length $\ell$, during which $i$ packets are inserted by the adversary, will result in a decrease to the potential of at least $\ell - i(1 + 5/\ln \threshold)$.
\label{lem:success}
\end{lemma}

\begin{proof}
First, let us look at the change in potential due to the inserted packets.  The first term $N_t$ increases by $i$ because there are $i$ new packets; the last term increases by $5i/\ln \threshold$ because these packets are inactive (until the next silent epoch).  The other terms in the potential function are not affected by new packet injections.

Now consider the change in potential due to the successful epoch.   The first term $N_t$ will decrease by $\ell$, as the epoch corresponds with a decoding window (\lemref{decodingwindow}) of size $\ell$ and hence $\ell$ packets are successfully delivered. 
The net change in the second term is non-positive, as the contention $c_t$ does not increase upon a successful epoch (since no broadcast probabilities are increased). Similarly, the net change in $s(t)$ is also non-positive, as $\pmin$ cannot decrease upon a successful epoch. (It may increase if the packet(s) with minimal joining probability are delivered.) 
Finally, $u(t)$ does not change since inactive packets are unaffected by a successful epoch. 
\end{proof}

Having established preliminaries regarding the effects of packet arrivals, packet activations, and successful epochs on $\Phi(t)$, we now examine the behavior of our potential function $\Phi(t)$.

\begin{lemma}
An error epoch increases $\Phi(t)$ by at most $\threshold + i(1 + 5/\ln \threshold) + 2$, where $i$ is the number of packets that arrive during the epoch.
\mab{we may want to restate the lemma to separate the case of error epochs of length one and error epochs of length $\kappa$. This may tighten the bounds in the proof.}
\label{lem:error_ok}
\end{lemma}
\begin{proof}
There are two types of error epochs: a silent epoch with high contention or an overfull epoch without high contention.  In both cases, the $i$ new packets that arrive increase the first term of the potential by $i$ and the last term of the potential by $5i /\ln \threshold$.
\begin{itemize}
    \item When there is an overfull epoch, $\pmin$ decreases multiplicatively by a factor of $\threshold^{1/4}$, causing an increase to $s(t)$ of exactly $1$.  Changes to $c_t$ only decrease the potential.    There is no change in potential due to packet activations because the epoch is not silent.  
    
    \item When there is a silent epoch, $c_t$ increases multiplicatively by a factor of $\threshold^{1/4}$, causing an increase to $\mbox{logC}(t)$ of exactly $\threshold$.  Changes to $\pmin$ only decrease the potential. By \lemref{activation}, the packet activations caused by silent (error) epoch increase $\Phi(t)$ by at most $2$.  
\end{itemize}
Thus the overall increase to $\Phi(t)$ caused by an error epoch is at most $\threshold + 2$.

\end{proof}

\begin{lemma}
Consider a given non-error epoch of length $\ell$, during which $i$ packets are inserted by the adversary. If any of the following conditions hold: (i) $\Phi(t) > 6\threshold$ or (ii) $\pmin < 1/\sqrt{\threshold}$ or (iii) contention $\geq \threshold^{1/4}$, then $\Phi(t)$ decreases by at least $\ell (1 - 1 /\threshold) - i(1 + 5 / \ln \threshold)$. Otherwise, $\Phi(t)$ decreases by at least $\ell (1 - 1/\threshold) - i(1 + 5 / \ln \threshold) - 2$.
\label{lem:non-error_good}
\end{lemma}
\john{$(1 - 1/\threshold)$ could be replaced with $\min(\ell, \threshold - 1)$, which would be a slightly stronger (and still correct) bound.}

\begin{proof}
Note that if the epoch is a successful epoch, we can immediately apply \lemref{success} to achieve the desired result. Thus we only need to consider overfull and silent epochs.  Notably, if contention $\geq \threshold^{1/4}$ and $\leq \threshold^{3/4}$, then the epoch is successful (because it is a non-error epoch).

By \lemref{arrivals}, we know that the $i$ packet arrivals will increase $\Phi(t)$ by exactly $i(1 + 5/\ln \threshold)$. By  \lemref{activation}, we know that, if $\Phi(t) \geq 6\threshold$ or $\pmin \leq 1/\sqrt{\threshold}$, any packet activations that occur during this epoch do not increase $\Phi(t)$.  If contention $\geq \threshold^{1/4}$, then the epoch is not silent (because it is a non-error epoch) and so no packets are activated.  Otherwise, the packet activations increase $\Phi(t)$ by at most $2$.

Thus, it is sufficient to show that, setting aside packet arrivals and activations, $\Phi(t)$ decreases by at least $\ell(1 - 1/\threshold)$ over this epoch. We proceed by casework, considering high contention, low contention, and contention near $c_* = \sqrt{\threshold}$.

\paragraph{Case 1: \boldmath$c_t > \kappa^{3/4}$.}

Since the epoch is not an error epoch, it cannot be a silent epoch. Thus we need only consider the case when it is an overfull epoch.

An overfull epoch is of length $\ell = \threshold$ by definition. Overfull epochs have no effect on the first and fourth terms in $\Phi(t)$. The epoch will cause a decrease to $\mbox{logC}(t)$ of $\threshold$ due to a multiplicative decrease in $c_t$ by a factor of $\threshold^{1/4}$. It will also increase $s(t)$ by $1$ due to $\pmin$ decreasing multiplicatively by a factor of $\threshold^{1/4}$.
Thus we have a net decrease in $\Phi$ of $\threshold - 1 = \ell(1 - 1 /\threshold)$.

\paragraph{Case 2: \boldmath$c_t < \kappa^{1/4}$.}

Since the epoch is not an error epoch, it cannot be an overfull epoch. We have already considered successful epochs, so we need only consider a silent epoch.

A silent epoch has no effect on the first and last terms in $\Phi(t)$, and the second term, $\mbox{logC}(t)$, will be $0$ both before and after this epoch due to contention remaining below $c_*$. Thus the only term affected by this epoch is $s(t)$, which will decrease by $1 = \ell$ as $\pmin$ increases multiplicatively by a factor of $\threshold^{1/4}$.

\paragraph{Case 3: \boldmath$\kappa^{1/4} \leq c_t \leq \kappa^{3/4}$.}

Because the epoch is not an error epoch, it must be a successful epoch. As noted earlier, we can then apply \lemref{success} to complete the proof.

\end{proof}

\subsection{Bounding the Potential}

\begin{lemma}
Suppose that in every window of size $w$, for any given $w \geq 16\threshold^2$, there are at most $\big(1 - \frac{5}{\ln \threshold}\big)w$ packet arrivals.  Then $\Phi(t) \leq 2w$, with probability at least $1 - 1/\poly{w}$.
\label{lem:throughput}
\end{lemma}
\begin{proof}
Consider the last time before $t$ when $\Phi$ was less than  $6\threshold$. We label this time $t'$.
The interval $(t', t]$ is of length $t - t'$, so at most 
$\big\lceil \frac{t - t'}{ w} \big\rceil 
\big(1 - \frac{5}{\ln \threshold}\big) w < (t - t' + w) \big( 1 - \frac{5}{\ln \threshold}\big)$ packets arrive during interval $(t', t]$. 
Furthermore, $\Phi> 6\threshold$ over this interval.

Suppose that $j$ error epochs occur during $(t', t]$, for any arbitrary $t'$.  By \lemref{intervalerror}, we know that $j \leq \sqrt{w} + c(t_2 - t_1) / 2^{\Theta(\threshold^{1/4})}$, for a fixed constant $c$, with probability $\geq 1 - 1/ \poly{w}$. For the remainder of the proof, we assume that this event occurs.

We observe that at least $(t - t') - j\threshold - \threshold$ time was spent in non-error epochs that complete by time $t$. (Note that there is potentially one incomplete epoch at the end, which has at most $\kappa$ slots.)

Let $\ell$ be the number of slots in completed non-error epochs in the interval $(t', t]$.  Let $i$ be the number of packets that are injected during that interval.  Applying  \lemreftwo{error_ok}{non-error_good}, it follows that:
\begin{align*}
\Phi(t) &\,\leq\, \Phi(t') - \ell\left(1 - \frac{1}{\threshold}\right) + i\left(1 + \frac{5}{\ln \threshold}\right) + j(\threshold + 2)\\
        &\,\leq\, 6\threshold - \big(t - t' - (j+1)\threshold\big)\left(1 - \frac{1}{\threshold}\right) + \\ 
        & \qquad\qquad\qquad (t - t' + w) \left(1 - \frac{5}{\ln \threshold}\right) \left(1 + \frac{5}{\ln \threshold}\right) + j(\threshold + 2) \\
        &\,\leq\, 7\threshold - (t - t')\left(1 - \frac{1}{\threshold}\right) + (t - t' + w)\left(1 - \frac{25}{\ln^2\threshold}\right) + 2j(\threshold + 2) \\
        &\,\leq\, 7\threshold + w - (t - t')\left(\frac{24}{\ln^2 \threshold}\right) + 2j(\threshold + 2)\,.
\end{align*}

Thus, if $\Phi(t) \geq 2w$ then:
\begin{align*}
    2j(\threshold + 2) \geq w - 7\threshold + (t - t')\left(\frac{24}{\ln^2 \threshold}\right) 
\end{align*}
and hence:
\begin{align*}
    j & > w/(\threshold+2) - 7 +  (t - t')\left(\frac{12}{\threshold \ln^2 \threshold}\right) \\
      & \geq \sqrt{w} + (t - t')\left(\frac{12}{\threshold \ln^2 \threshold}\right).
\end{align*}
This, however, contradicts our assertion that $j \leq \sqrt{w} + c(t_2 - t_1) / 2^{\Theta(\threshold^{1/4})}$, for a fixed constant $c$, implying that $\Phi(t) < 2w$.
\end{proof}

\begin{theorem}
Suppose that in every window of size $w$, for any given $w \geq 16\threshold^2$, there are at most $\big(1 - \frac{5}{\ln \threshold}\big)w$ packet arrivals.
Then at any arbitrary time $t$, the number of packets in the system is at most~$2w$ with probability at least $1 - 1/\poly{w}$.
\label{thm:throughput}
\end{theorem}
\begin{proof}
    By \lemref{throughput}, we know that $\Phi(t) < 2w$, with probability at least $1 - 1/\poly{w}$.  Since the first term of the potential is the number of packets in the system, this immediately implies that that are at most $2w$ packets in the system.
\end{proof}



\section{Latency Analysis}
\label{sec:latency}

In this section, we consider the maximum latency of a packet.  Under the assumption that the adversary can inject $(1 - o(1))w$ packets in a window of size $w$, it is unavoidable that the worst-case latency of a packet will be $\Omega(w)$: if the adversary injects $w$ packets in a slot, and the maximum capacity is one packet per slot, then at least half the packets will take at least $w/2$ slots.  Our goal in this section is to show that the worst-case throughput is $O(\sqrt{\threshold} w \ln^3 w)$.  We also consider the batch case and show that if $n$ packets arrive simultaneously, they will all be delivered by time $n + O(1)$.

\subsection{Packet Latency Analysis}

We begin by defining a \defn{sparse system} event that occurs when there are very few active packets in the system.  We will show that whenever a sparse system event occurs, each packet has a reasonable chance of success.
\begin{definition}
A \defn{sparse system} event occurs in slot $t$ if: (i)~$\Phi(t) \leq 6\threshold$; 
(ii)~contention $< \threshold^{1/4}$, 
(iii)~$\pmin \geq 1/\sqrt{\threshold}$.
\end{definition}

When a sparse system event occurs, it means that every active packet was joining epochs with a relatively high probability; after the epoch completes, any new packets that activated have lowered $\pmin$.  Notice that when a sparse system event does not occur, then the potential decreases (aside from new packets being injected): \lemref{non-error_good} shows that potential decreases in non-error epochs, and \lemref{error_unlikely} shows that error epochs are unlikely.

We first show that no packet observes too many completed sparse system events before it is delivered.
\begin{lemma}
\label{lem:howmanysparse}
A given packet $p$ remains active in the system for at most $O(\sqrt{\threshold}\ln{w})$ sparse system events, with probability at least $1 - 1/\poly{w}$.
\end{lemma}
\begin{proof}
    When a sparse system event occurs, two facts are true: (i) the probability that packet $p$ joins that epoch is $\geq 1/\sqrt{\threshold}$ because of the condition on $\pmin$; (ii) the number of active packets is $\leq \threshold$.  The latter fact follows because if there were more than $\threshold$ active packets, then the total contention would be at least $\threshold \cdot 1/\sqrt{\threshold} \geq \sqrt{\threshold}$, contradicting the condition on contention.
    
    Since there are at most $\threshold$ active packets, then all the packets that join the epoch will be successfully delivered. Since $p$ (once active) joins with probability $1/\sqrt{\threshold}$, we conclude that $p$ had a $1/\sqrt{\threshold}$ probability of being successfully delivered.  The probability that $p$ fails to be delivered after $\Theta(\sqrt{\threshold}\ln{w})$ sparse system events (while it is active) is at most $(1 - 1/\sqrt{\threshold})^{\Theta(\sqrt{\threshold}\ln{w})} \leq 1/\poly{w}$. 
\end{proof}

We now argue that, with high probability, a sparse system event occurs every $O(w \ln w)$ slots.  The first step is to show that if the potential is some value $\alpha$ at a time $t$, then a sparse system event occurs within time $O((\alpha + w) \ln^2 w)$.
\begin{lemma}
Suppose that in every window of size $w$, for any given $w \geq 16\threshold^2$, there are at most $\big(1 - \frac{5}{\ln \threshold}\big)w$ packet arrivals.  Fix a time $t$ and define $\alpha = \Phi(t)$.  Then, with high probability in $w$, by time $t + 1.1(\alpha + w)\ln^2 w)$, either all packets active at time $t$ are delivered or a sparse system event occurs.  
\label{lem:decrease}
\end{lemma}
\begin{proof}
Let $t_1 = t$ and let $t_2 = t_1 + qw$ be the largest number that is congruent with $t_1 \pmod{w}$ and smaller than $t_1 + 1.1(\alpha + w)\ln^2 w$.  It is necessarily true then that $t_2 - t_1$ is a multiple of $w$ satisfying $t_2 - t_1 \geq t + (\alpha + w)\ln^2 w$ (i.e., relaxing the constant in the inequality from $1.1$ to $1$ allows us to take $t_2 - t_1$ to be a multiple of $w$).  Assume, for the sake of contradiction, that no sparse system event occurs during the interval $(t_1, t_2]$.  We will show that all the packets active at time $t_1$ are delivered with high probability.

Suppose that $j$ error epochs occur during $(t_1, t_2]$. By \lemref{intervalerror}, we know that $j \leq \sqrt{w} + 6(t_2 - t_1) / 2^{\Theta(\threshold^{1/4})}$, with probability $\geq 1 - 1/ \poly{w}$. For the remainder of the proof, we assume that this event occurs.

We observe that at least $(t_2 - t_1) - j\threshold - \threshold$ time was spent in non-error epochs that complete by time $t_2$. (Note that there is potentially one incomplete epoch at the end, which has at most $\kappa$ slots.)

We also observe that the number of packets that are injected in the interval $(t_1, t_2]$ is at most:
$\big\lceil \frac{t - t'}{w} \big\rceil 
\big(1 - \frac{5}{\ln \threshold}\big) w < (t_2 - t_1) \big( 1 - \frac{5}{\ln \threshold}\big)$. 

We can now bound the potential at time $t_2$.  At slot $t_1$, $\Phi(t_1) = \alpha$; by \lemref{non-error_good} that the potential decreases by $(1 - 1/\threshold)$ in each slot that was not part of an error epoch and increases $(1 = 5/\ln \threshold)$ for each injected packet; by \lemref{error_ok}, the potential increases at most $\threshold + 2)$ for each error epoch.  We use this to bound how long until the potential reaches zero. 

Let $\ell$ be the number of slots in completed non-error epochs in the interval $(t_1, t_2]$.  Let $i$ be the number of packets that are injected during that interval.  Applying  \lemreftwo{error_ok}{non-error_good}, it follows that:
\begin{align*}
\Phi(t_2) &\,\leq\, \Phi(t_1) - \ell\left(1 - \frac{1}{\threshold}\right) + i\left(1 + \frac{5}{\ln \threshold}\right) + j(\threshold + 2)\\
        &\,\leq\, \alpha - \big(t_2 - t_1 - (j+1)\threshold\big)\left(1 - \frac{1}{\threshold}\right) + \\ 
        & \qquad\qquad\qquad (t_2 - t_1) \left(1 - \frac{5}{\ln \threshold}\right) \left(1 + \frac{5}{\ln \threshold}\right) + j(\threshold + 2) \\
        &\,\leq\, \alpha - (t_2 - t_1)\left(1 - \frac{1}{\threshold}\right) + (t_2 - t_1)\left(1 - \frac{25}{\ln^2\threshold}\right) + (2j+1)(\threshold + 2) \\
        &\,\leq\, \alpha - (t_2 - t_1)\left(\frac{24}{\ln^2 \threshold}\right) + 6j\threshold \\
        &\,\leq\, \alpha - (t_2 - t_1)\left(\frac{24}{\ln^2 \threshold}\right) + 6\sqrt{w}\threshold + (t_2 - t_1)\frac{36\threshold}{2^{\Theta(\threshold^{1/4})}} \\
        &\,\leq\, (\alpha + 6w) - (t_2 - t_1)\left(\frac{18}{\ln^2 \threshold}\right)\,.
\end{align*}
Thus, since $(t_2 - t_1) \geq (\alpha + w) \ln^2 w \geq (\alpha + w) \ln^2 \threshold$, we can conclude that the potential has dropped below zero, implying that all the packets in the system have been delivered.
\end{proof}

A corollary of the previous lemma is that a sparse system event occurs in every interval of size $O(w \ln^2 w)$, with high probability.
\begin{corollary}
\label{cor:sparsewindow}
Suppose that in every window of size $w$, for any given $w \geq 16\threshold^2$, there are at most $\big(1 - \frac{5}{\ln \threshold}\big)w$ packet arrivals.  For any interval of time $(t_1, t_2)$ where $(t_2 - t_1) > 4w \ln^2 w$, then either every packet active at time $t_1$ is delivered by time $t_2$, or a sparse system event occurs in that interval, with probability $1 - 1/\poly{w}$.
\end{corollary}
\begin{proof}
From \lemref{throughput}, we know that at time $t_1$, $\Phi(t) \leq 2w$ with probability at least $1 - \poly{w}$.  From \lemref{decrease}, we know that by time $t_1 + 1.1(\alpha + w) \ln^2 w \leq t_1 + 4w \ln^2 w \leq t_2$, either all the packets active at time $t_1$ are delivered or a sparse system event occurs, with probability at least $1 - \poly{w}$. 
\end{proof}

We can conclude from the preceding lemmas that, with high probability, every packet is delivered in $O(\sqrt{\threshold} w \ln^3 w)$ time.
\begin{theorem}
\label{thm:latency}
Suppose that in every window of size $w$, for any given $w \geq 16\threshold^2$, there are at most $\big(1 - \frac{5}{\ln \threshold}\big)w$ packet arrivals. Then any given packet is delivered within time $O(w \sqrt{\threshold} \ln^3 w)$ with probability $1 - 1/\poly{w}$.
\end{theorem}
\begin{proof}
    Assume that packet $p$ is injected at time $t$.  We first need to consider how long until packet $p$ becomes active.  
    
    We now bound how long until a sparse system event occurs which leads to a non-error epoch; such an epoch will be silent, and hence packet $p$ will become active.  In every window of length $\Theta(w \ln^2 w)$, there is a sparse system event with probability $1 - \poly{w}$ or all active packets are delivered, by \corref{sparsewindow}.  If all active packets are delivered, then the next epoch is necessarily silent. If a sparse system event occurs, by \lemref{error_unlikely}, we know that the epoch comprising the sparse system event is an error epoch with probability $1 / 2^{\Theta(\threshold^{1/4})} \leq 1/e$; otherwise it is silent and $p$ becomes active.  Thus, the probability that there is no silent epoch (and hence packet $p$ remains inactive) for $\Theta(\ln w)$ intervals of length $\Theta(w \ln^2 w)$ is $\leq 1/\poly{w}$.
    
    Once $p$ becomes active, we know by \lemref{howmanysparse} that $p$ will be delivered within $O(\sqrt{\threshold} \ln w)$ sparse system events, with probability at least $1 - \poly{w}$.  We know by \corref{sparsewindow} that in every interval of length $O(w \ln^2 w)$, either all packets (including $p$) are delivered or there is a sparse system event, with probability at least $1 - 1/\poly{w}$.    Thus, within $O(w\sqrt{\threshold} \ln^3 w)$, packet $p$ is delivered with probability at least $1 - 1/\poly{w}$.
\end{proof}

\subsection{Batch Analysis}

\begin{theorem}
\label{thm:batch:latency}
Suppose that a batch of $n \neq 0$ packets arrives at time $0$, and the adversary does not inject any packets besides this initial batch. Then with high probability in $n$, all of the packets have successfully been delivered by time $n(1 + 10/\threshold) + O(\threshold)$.
\end{theorem}

\begin{proof}
Assume $n$ packets arrive at time $0$, which is a silent epoch, and then all immediately activate. Thus at time $1$ we have 
\begin{align*}
\Phi(1) &= n + \max \left( 0, 4\threshold \log_\threshold \frac{n}{\threshold} \right) + 4\log_\threshold \frac{1}{1/\sqrt{\threshold}} \\
&\leq n(1 + 5/\ln \threshold) + 2
\end{align*}

We now consider $\Phi(t)$. At least $t - 1 - \threshold$ slots have been in non-error epochs since time $1$, and let $j$ be the number of error epochs. Then

\begin{align*}
\Phi(t) &\leq \Phi(1) - (t - 1 - \threshold)(1 - 1/\threshold) + j(\threshold + 2) \\
&\leq \Phi(1) - t(1 - 1/\threshold) + (j+1)(\threshold + 2).
\end{align*}

With high probability in $n$, we have, via an application of Chernoff Bounds, that $j \leq 2t / 2^{\Theta(\threshold^{1/4})}$, for any $t \geq n$. We assume this holds to obtain

\begin{align*}
\Phi(t) &\leq \Phi(1) - t(1 - 1/\threshold) + (1 + 2t/2^{\Theta(\threshold^{1/4})})(\threshold + 2).
\end{align*}

Taking $t = 4\threshold + n(1 + 10/\threshold)$, we have that $\Phi(t) \leq 0$. Thus there are, with high probability in $n$, no packets left in the system by slot $t = 4\threshold + n(1 + 10/\threshold)$, which is $O(n(1 + 1/\threshold))$.
\end{proof}


\begin{acks}
This research was supported by NSF grants
CCF-2118832,  
CCF-2106827, 
CSR-1763680,  
CCF-1716252, 
CNS-1938709, and  
CCF-1725543, 
as well as by Singapore MOE grant MOE2018-T2-1-160 (Beyond Worst-Case Analysis).
\end{acks}

\bibliographystyle{ACM-Reference-Format}
\balance
\bibliography{backoff,bender-bib}


\end{document}